\documentclass[12pt]{spieman}  
\usepackage{amsmath,amsfonts,amssymb}
\usepackage{graphicx}
\usepackage{setspace}
\usepackage{tocloft}
\usepackage[utf8]{inputenc}
\usepackage[left]{lineno}
\usepackage[dvipsnames]{xcolor}

\title{Curved Holographic Optical Elements from a Geometric View Point}

\author[a,*]{Tobias Graf}
\affil[a]{Robert Bosch GmbH, Corporate Research and Advance Development, Robert-Bosch-Campus 1, 71272 Renningen, Germany}

\cftpagenumbersoff{figure}
\cftpagenumbersoff{table} 

\newtheorem{proposition}{Proposition}
\newtheorem{problem}{Problem}
\newtheorem{remark}{Remark}

\begin{document} 
\maketitle
\begin{abstract}
	Motivated by recent developments in augmented reality display technology and notably smartglasses, we present a geometric framework to model the optical effects of deformations of planar holographic optical elements (HOE) into curved surfaces. In particular, we consider deformations into shapes which do not preserve the Gaussian curvature, such as deforming planar regions into sphere segments. 
\end{abstract}


{\noindent \footnotesize\textbf{*} \linkable{tobias.graf3@de.bosch.com} }

\begin{spacing}{1}   


\section{Introduction}
\emph{Augmented reality} (AR) is the technology to augment a user's view of the real world with additional virtual objects. Ideally, this augmentation is to be achieved without obstructing the view of the real world while placing the virtual objects in the user's field of view (FoV). Well known examples of AR displays are \emph{head-up displays} (HUDs), \emph{helmet-mounted displays} (HMD), and \emph{head-worn displays} (HWD). Interestingly, to the best of the author's knowledge, head-up displays were also one of the first industrial applications of \emph{holographic optical elements} (HOE) \cite{close_holographic_1975}. HOEs offer multiple interesting design aspects, which make them suitable, if not necessary for highly transparent, small form-factor, unobstructive see-through AR displays, in particular HUDs, HMDs or HWDs. HOEs are also used as combiner elements in \emph{retina scanner displays} (RSDs) for smartglasses\cite{akutsu_compact_2019}  \cite{viirre_virtual_1998}, since they allow manipulation of the light emitted from a projection engine without affecting the real-world view in a see-trough AR system due to their wavelength selectivity.

As already indicated, in order to realize highly transparent, small form-factor, unobstructive see-through AR displays, optical engineers have moved beyond classical components such as mirrors (reflective elements) and lenses (refractive elements). Modern system architectures for AR-displays, e.g. Microsoft's first generation HoloLens \cite{kress_optical_2017}, often rely also on diffractive optical elements such as surface relief gratings, to manipulate light propagation in new ways. Moreover, a planar waveguide combiner with surface relief gratings for in- and out-coupling can also be combined with MEMS (micro-electro-mechanical-system) mirror based projection units \cite{ayras_near--eye_2010}, as in Microsoft's second generation HoloLens. However, diffractive optical elements based on surface relief structures affect light of all wavelengths and their performance may suffer from chromatic effects as well as undesirable higher diffraction orders. Therefore, holographic optical elements based on volume diffraction (also called thick holograms \cite{kogelnik_coupled-wave_1969}) which can reach high diffraction efficiencies in combination with high angle or wavelength selectivity have been investigated early on as a promising technology for AR displays\cite{close_holographic_1975}.

The interest in holographic optical elements (HOE) has recently increased even further, as new holographic materials have been developed \cite{kowalski_design_2016}. In recent years, improved photopolymers for holographic applications have undergone significant improvements and are now commercially available \cite{jurbergs_new_2009}. Due to the easier handling and storage they have become an attractive alternative to classical holographic materials such as dichromate gelatine. Moreover, modern photopolymers lend themselves to industrial replication processes by including the photopolymer in a foil stack of carrier and cover substrate material. The thin holographic foil can then be rolled up for storage. Handling during HOE manufacturing can also be simplified through automated equipment which can process high volumes of HOE e.g. in roll-to-roll replication processes\cite{bjelkhagen_mass_2017}. Thanks to the advancements of photopolymers in recent years, the optical designer's toolbox has been expanded even further. Photopolymers offer a holographic material which can be applied to planar and curved carrier substrates. Thus, the optical designer can employ the characteristic wavelength and angle selectivity of volume holograms in HOEs to control light propagation in AR systems. Due to the aforementioned characteristics, HOEs offer multiple interesting design aspects, which make them promising candidates for highly transparent, small form-factor, unobstructive see-through AR displays, in particular smartglasses. 

Moreover, HOEs recorded in photopolymers can be embedded in prescription lenses for eye-wear using high volume industrial processes compatible with existing eye-wear standards, processes, and equipment\cite{korner_casting_2018}. 

Therefore, RSDs based on light projection engines using MEMS mirrors in combination with holographic combiners are promising candidates for an upcoming generation of smartglasses\cite{akutsu_compact_2019, viirre_virtual_1998, kollin_optical_1995}.

This paper is motivated by the recent efforts to develop head-worn AR displays, in particular smartglasses, for consumer and industrial applications. The need to account for vision correction with prescription lenses is expected to be a particularly crucial aspect of smartglasses technology for consumer applications.

We present a geometric framework, which serves as a building block for simulation and design methods for a system architecture of smartglasses, especially RSDs. The primary focus in this paper is a modeling approach which can account for deformations of HOE after the recording process.  In particular, we consider deformations which do not preserve the Gaussian curvature\cite{carmo_differential_2018}. This important aspect distinguishes, among others, the present discussion from previously published work on curved HOE \cite{bang_curved_2019}, where the authors present simulation and experimental results for curved HOEs in the shape of cylinder segments, which  intrinsicially have Gaussian curvature zero since only one principal curvature does not vanish and the Gaussian curvature can be expressed of the product of the principal curvatures \cite{carmo_differential_2018}.  We would like to point out that the framework discussed here was developed independently and can be modified to include cylindrical and more general deformations as well by chosing suitable sets of moving frames. We plan to return to this in a separate investigation. The method of moving frames is introduced in Section \ref{sec:fw-prblm} to combine what we will refer to as as macroscopic (i.e. surface) geometry with microscopic (i.e. volume grating) geometry in a single modelling framework. Moving frames are known in mathematics for example in the study of curves and surfaces in classical differential geometry\cite{cartan_differential_2006, carmo_differential_2018} as well as in image analysis\cite{sapiro_geometric_2006}. In this paper, we develop a geometric framework which also uses fields of frames and is focused on the application in numerical algorithms for optical engineering of RSDs. We also would like to point out that in the present paper, we are not investigating the geometry of the wavefronts diffracted at the HOE as has been don in earlier work\cite{verboven_aberration_1986, peng_nonparaxial_1986} in order to derive aberration coefficients with respect to some reference wavefront, if the HOE is recorded on a curved substrate directly. Instead we focus on the question, how the optical characteristics of an HOE are affected, if its macroscopic geometry is changed after recording. 

\subsection{Statement of the Problems to be solved}
The two principal problems addressed in this paper are as follows:

\subsubsection{Forward Problem of curved HOE} \label{sssec:fw}
\begin{problem} \label{pblm:fw} 
	Given a planar holographic optical element (HOE) $\mathcal{H}$ and a prescribed in advance surface $\mathcal{S}$, what are the optical characteristics of the HOE $\tilde{\mathcal{H}}$ obtained after deforming $\mathcal{H}$ into the shape $\mathcal{S}$?     	
\end{problem}

The illustration in Figure \ref{fig:prblm-fw} is a visualization of Problem \ref{pblm:fw} in the case were a planar HOE is deformed after recording, e.g. into a segment of a sphere. The forward problem aims to predict the diffraction behavior of the deformed HOE knowing the properties of the recorded HOE. 

\begin{figure}
	\centering
	\includegraphics[scale = 0.8]{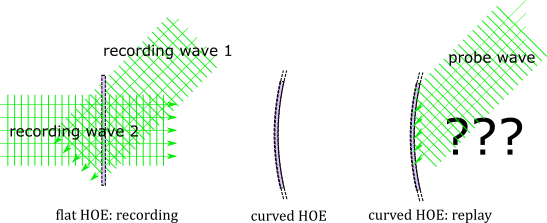} 
	\caption{Forward problem of HOE deformation: After recording an HOE with two coherent recording waves (often referred to as reference and object waves in the holography literature), the macroscopic geometry is changed. The forward problem consists of the task to predict the diffraction of a probe wave at the deformed HOE.} \label{fig:prblm-fw}
\end{figure}

\subsubsection{Inverse Problem of curved HOE} \label{sssec:inv}
\begin{problem} \label{pblm:inv} 
	Given a surface $\mathcal{S}$, prescribed in advance probe wave $\mathbf{E}_{p}$  (or familiy of reference wavefronts, e. g. for polychromatic systems) in combination with a prescribed in advance target wavefront $\mathbf{E}_{d}$ (or family of target wavefronts), find a holographic optical element (HOE) $\mathcal{H}$ such that the curved HOE resulting from reshaping  $\mathcal{H}$ into the geometry of $\mathcal{S}$ results in a HOE which transforms $\mathbf{E}_{p}$ into $\mathbf{E}_{d}$.	 
\end{problem}

In Figure \ref{fig:prblm-inv} the inverse problem, i.e. Problem \ref{pblm:inv}, is illustrated in the case were a planar HOE is deformed after recording, e.g. into a segment of a sphere. The deformed HOE is expected to have certain diffraction properties, e.g diffracting a prescribed probe wave into a prescribed diffracted wave. The goal of the inverse problem is to determine an undeformed HOE structure such that after deformation into the reconstruction geometry the deformed HOE shows the prescribed in advance diffraction behavior. 

\begin{figure}
	\centering
	\includegraphics[scale = 0.8]{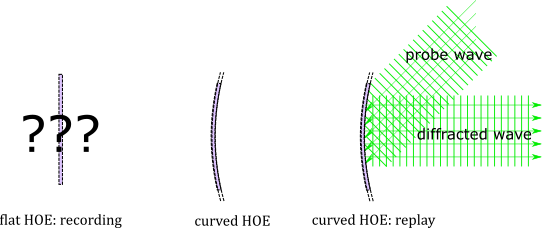}
	\caption{Inverse problem of HOE deformation: Given a prescribed in advance macroscopic recording and reconstruction geometry as well as a prescribed in advance diffraction behavior of the HOE during reconstruction (i.e. its micrsocopic structure), the inverse problem consists of the task to determine the microscopic geometry of the HOE in the recording geometry prior to deformation into the reconstruction geometry.} \label{fig:prblm-inv}
\end{figure}

\subsection{Outline of the paper}
The paper is organized as follows:
\begin{itemize}
	\item In Section \ref{sec:RSDandHOE} we briefly introduce the core system architecture of a retina scanner display (RSD) and discuss some important aspects of industrial replication of HOEs and the embedding of HOEs in ophtalmic lenses.
	\item In Section \ref{sec:holo} we give a cursory introduction to volume holography, starting with the example of a volume grating obtained by interference of two plane waves during the recording process. Based on the invesitgation of the volume grating we introduce the notion of a grating vector. With this background on volume gratings, we turn our attention to HOE obtained from the interference of spherical waves and define the HOE model which is used throughout this paper.
	\item In Section \ref{sec:fw-prblm} we define a space of admissible HOE deformations and investigate some of their properties to formalize Problem \ref{pblm:fw} and provide a geometric framework for a solution. We then turn our attention to Problem \ref{pblm:inv} using the previously developed methods. 
	\item In Section \ref{sec:applications}, we show results of numerical experiments using a basic Matlab implementation of the core principles of the methods developed in the previous sections. We also return to the retina scanner display as the motivating application and discuss the effects of mechanical deformations of  an HOE converting a divergent spherical wave into a convergent spherical wave.
	\item We conclude our presentation with a summary of our results and an outlook of on-going and future research in Section \ref{sec:conclusion}.
\end{itemize}

\section{System architecture and manufacturing considerations for HOE combiners in retina scanner displays} \label{sec:RSDandHOE}
In this section, we briefly introduce the basic system architecture of a retina scanner display (RSD) and discuss some aspects of industrial HOE replication and HOE embedding in ophtalmic lenses. This presentation provides the motivation for the following mathematical discussions and serves to clarify why certain assumptions are being made during the derivation of the method. 

\subsection{Fundamental system architecture of retina scanner displays}  \label{subsec:RSD}
We briefly recall the RSD architecture and the use of HOEs as combiners in such a system\cite{akutsu_compact_2019}. The conept of (virtual) retina scanner displays goes back at least to the 1990s\cite{viirre_virtual_1998, kollin_optical_1995}. However, only recent advances in MEMS technology, laser integration, holographic materials, holography, and lens casting have put the concept within reach of industrialization. Figure \ref{fig:RSDbasics} shows a schematic illustration of an RSD system. A projection unit, typically containing a laser module and a MEMS mirror module, scans the projected light across the lenses. HOEs attached to or embedded in the lenses diffract the light of the projection unit towards the user's eyes without perturbing the incident light from the surrounding environment due to wavelength selectivity of the HOEs\cite{kick_sequential_2018}.  

\begin{figure}
	\centering
	\includegraphics[scale = 1.0]{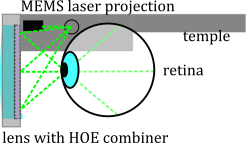}
	\caption{Schematic illustration of a retina scanner display (RSD). A holographic optical element (HOE) in the lens is used as a combiner to redirect light from a MEMS (micro-electro-mechanical-system) -mirror based projection unit integrated in the temple towards the user's eye. By modulating the light sources time sequentially depending on the scanning mirror position, local visual stimuli can be induced in the user's eye at the retina.}\label{fig:RSDbasics}
\end{figure}

\subsection{Industrial replication of holographic optical elements} \label{subsec:HOEreplic}
While it is possible to record an HOE on curved substrate, it may be advantageous or even necessary to record HOEs in a planar setting, even in cases where the application profits from curved HOE. This may for example be due to the manufacturing process such as a roll-to-roll replication process targeting mass production of HOE at industrial scales \cite{bjelkhagen_mass_2017}. The aforementioned process is based on the replication of a \emph{master} HOE located on a planar carrier substrate using single beam illumination of a photopolymer foil laminated onto the master. The photolpolymer foil is supplied from a roll and is laminated automatically to the master before replication and delaminated automatically and rolled up again afterwards, 

\subsection{Embedding HOE into prescription lenses} \label{subsec:HOEembed}
One motivation to investigate the deformation of curved HOE after recording is provided by the need to introduce prescription lenses for vision correction into a RSD. The compatibility with vision correction is expected to be an important aspect for smartglasses targeted at consumer applications. Moreover, there already exist industrial processes to embed HOE into polymer lenses\cite{korner_casting_2018} compatible with current standards in the ophtalmic industry. However, the process requires the HOEs to be shaped into sphere segments prior to embedding. Moreover, the casting material is cured with UV light and thus requires that the HOEs are already recorded and fixed in the photopolymer prior to embedding.

Aside from the workflow of the process presented in \cite{korner_casting_2018}, the thickness (and thus the weight) of the lens is another design aspect which favors curved HOE over planar ones. For a prescribed in advance lens diameter (in particular the relatively large diameters of the lens blanks compared to the final lens diameter after edging and fitting in the spectacles frame), the thickness of the lens is determined by the desired optical power and outer radii according to the lens maker equation\cite{hecht_optics_2002}

\begin{equation} \label{eq:lensmaker}
\frac{1}{f} = (n-1)\left( \frac{1}{R_1} - \frac{1}{R_1} + \frac{(n-1)d}{n R_1 R_2} \right)
\end{equation}

where $R_1$ and $R_2$ are the radii of curvature, $n$ the index of refraction of the lens material, and $d$ the center thickness of the lens. Note that we assume the radii in \eqref{eq:lensmaker} to be positive since we are primarily interested in convex-concave (or meniscus) lenses for vision correction. Moreover, we assume the lens to be embedded in air and thus use a refractive index of $n_0 = 1.0$ for the surrounding medium. Figure \ref{fig:embedding} shows in the left column two convex-concave lenses for vision correction, the top lens corresponding to a corrective lens for hyperopia, and the bottom lens is a negative lens to correct myopia. The central column illustrates the embedding of a planar HOE in the same type of lenses. We observe, that for a prescribed lens radius, the HOE plane introduces a lower bound for the lens thickness. Moreover, the embedding material thickness between HOE and the environment on the eye-side at the lens center is also limited by the overall lens thickness. In the right column, we see an HOE which was curved into a sphere segment prior to embedding. This allows a thinner lens compared to the planar HOE or an increased minimal thickness of the lens material for improved mechanical protection.

\begin{figure}
	\centering
	\includegraphics[scale = 1.0]{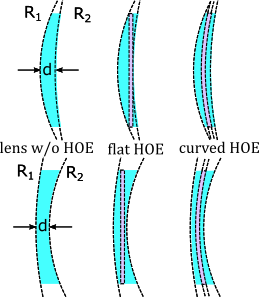}
	\caption{The illustrations show: meniscus lenses for vision correction (left column), with planar HOEs embedded (center column), and curved HOEs embedded (right column). Note that in case of ophtalmic lenses the world-side is to the left of the lens, the eye-side to the right.} \label{fig:embedding}
\end{figure}

\section{Fundamentals of holography applied to holographic optical elements} \label{sec:holo}

We begin with an investigation of the microscopic structure, i.e. the index modulation within a phase hologram, of two classes of holographical optical elements. First, we consider an HOE constructed by interference of two plane waves. We will explicitely compute the interference pattern and motivate the concept of the so-called grating vector. Moreover, the plane wave HOE will later serve as a local approximation for more general HOEs. One such class of more general HOEs will be discussed next, when we consider HOEs recorded by interference of two spherical waves. One important observation is the fact that in contrast to the plane wave HOE the grating vectors of the spherical wave hologram are not all collinear.
\subsection{Example: HOE constructed from plane wave interference (volume grating)} \label{subsec:planeexpl}
We begin our disucssion with the following holographic recording setup: the output of a coherent light source (typically a laser) is collimated and the resulting beam is divided into two light paths. The two beams are considered here for simplicity to represent ideal plane waves. A planar substrate carrying a photopolymer foil is placed in the region where these two plane waves interfere. The intensity of the interference pattern is translated by chemical reactions into a refractive index modulation in the photopolymer and can be fixed using UV-light. If the resulting HOE is illuminated with one of the recording waves serving as a probe wave, the light diffracted at the HOE strucuture will propagate in the same direction as the second recording wave. The angles of incidence of the waves at the photopolymer with respect to the surface normal of the substrate and a reference point (e.g. the center of the future HOE) are typically used to describe the recording geometry of such plane wave converting HOEs.

Figure \ref{fig:planarHOE} illustrates the above situation for a planar photopolymer. The photopolymer is illuminated by the two interfering waves and a refractive index modulation in the material occurs. After this \emph{recording} step -- and typically an additional fixation step -- the HOE can be \emph{reconstructed} (or \emph{replayed}). For an ideal HOE, if the probe wave used to reconstruct the hologram corresponds to one of the recording waves or their complex conjugates, the HOE is said to be reconstructed in on-Bragg condition. This is typically where diffraction efficiency, i.e. the ratio of the power in the diffracted wave and the incident wave, is at its maximum. Any deviation from the on-Bragg condition, e.g. changes in the angle of incidence or the wavelength, will result in a different diffraction angle and a reduced diffraction efficiency\cite{kick_sequential_2018}.

\begin{figure}
	\centering
	\includegraphics[scale=0.8]{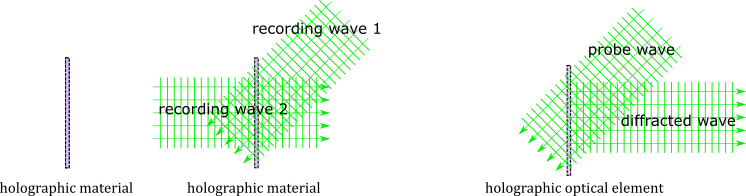} 
	\caption{Left: a planar photopolymer; Center: recording a volume grating via interference of two plane waves; Right: reconstruction in on-Bragg condition} \label{fig:planarHOE}
\end{figure}

We consider two plane waves $E_l = A_l(\mathbf{r}) \exp(\mathrm{i} \mathbf{k}_l \cdot \mathbf{r}), l = 1, 2$ with real amplitudes $A_l$.

The interference pattern of these two waves is given by 
\begin{equation} \label{eq:intensityvolgrating}
\begin{split}
\left| A_1(\mathbf{r}) \mathrm{e}^{\mathrm{i} \mathbf{k}_1 \cdot (\mathbf{r})} + A_2(\mathbf{r}) \mathrm{e}^{\mathrm{i} \mathbf{k}_2 \cdot (\mathbf{r})}\right|^2 = & \left| A_1 \right|^2 + \left| A_2 \right|^2 \\
& + 2A_1 A_2 \cos((\mathbf{k}_2 - \mathbf{k}_1)\cdot \mathbf{r})
\end{split}
\end{equation}

We assume thath the cosine modulated intensity pattern in \eqref{eq:intensityvolgrating} is tranferred linearly into a refractive index modulation. Note that $\cos((\mathbf{k}_2 - \mathbf{k}_1)\cdot \mathbf{r})$ is constant, if $(\mathbf{k}_2 - \mathbf{k}_1)\cdot \mathbf{r}$ is constant. Moreover, the condition
\begin{equation} \label{eq:BraggPlanes}
(\mathbf{k}_2 - \mathbf{k}_1)\cdot \mathbf{r} = \mathrm{const}
\end{equation}
is the representation of a plane (or familiy of parallel planes for varying constants). The normal vector is colinear to the vector
\begin{equation} \label{eq:kg}
\mathbf{k}_g = \mathbf{k}_2 - \mathbf{k}_1
\end{equation}
Since the cosine is $2\pi$--periodic, the next plane in direction of $\mathbf{k}_g$, where the cosine value is repeated, is located at a distance
\begin{equation} \label{eq:vgratingperiod}
\Lambda = \frac{2\pi}{|\mathbf{k}_g|
}
\end{equation}
The familiy of planes corresponding to fixed constant values $c_0 + l \Lambda, c_0 \in \mathbb{R}, l \in \mathbb{Z}$ in \eqref{eq:BraggPlanes}, are often refered to as \emph{Bragg planes}. They are all parallel and spaced equidistantly with the volume grating period $\Lambda$ as defined in \eqref{eq:vgratingperiod}.

Thus, the vector $\mathbf{k}_g$ encodes the orientation and the spacing of the Bragg planes and is called the \emph{grating vector}. 

Note that the grating vector is constant in this example and the resulting Bragg strucuture is periodic. Therefore, we will also refer to this type of HOE as a \emph{volume grating}. 

In the following, we will use the above observations to briefly introduce methods to determine the direction of propagation of a diffracted wave using vector formalisms. 

\begin{remark}
	Note that if we laminate the photopolymer on a curved substrate prior to the recording process, i. e. we record in curved holographic material as illustrated in Figure \ref{fig:curvedHOE}, then the Bragg structure is recorded in a curved HOE, however, the grating vectors are still all colinear. 
\end{remark}

\begin{remark}
	We have not considered the possibility of amplitude variations in the above discussion, to keep the presentation simple. Spatially varying amplitudes will however result in corresponding variations of the fringe contrasts in the interference and thus in the index modulation pattern\cite{uchida_calculation_1973}.
\end{remark}

\begin{figure}
	\centering
	\includegraphics[scale=0.8]{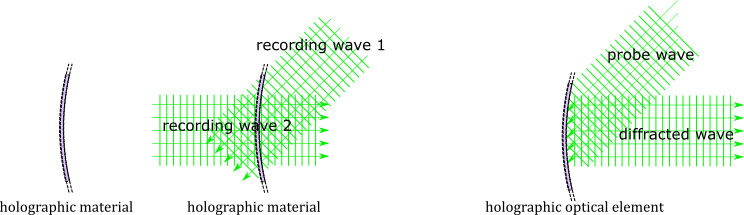}
	\caption{Illustration of the recording and on-Bragg replay process of a volume grating recorded in a curved photopolymer.} \label{fig:curvedHOE}
\end{figure}

\subsection{k-vector formalism and diffraction theories for volume gratings} \label{subsec:HOEtheory}
We briefly review some modelling aspects regarding the diffraction of light at a volume hologram, in particular the direction of propagation of the diffracted light (also refered to as the diffraction order), and the diffraction efficiency. At this point it is also important to point out that the geometric framework, which we are presenting in the following, is quite flexible and can be combined with 
different modelling approaches with varying complexity, depending on the physical accuracy required.

\begin{figure}
	\centering
	\includegraphics[scale=0.5]{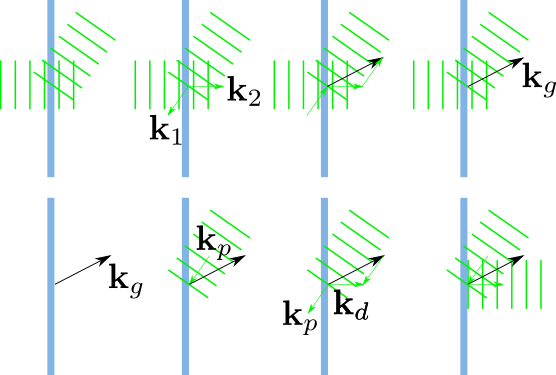}
	\caption{Local grating vector description based on plane wave approximation: The interference of plane waves, results in a periodic Bragg structure, which can be encoded in a single vector, the so-called grating vector.} \label{fig:kvector}
\end{figure}

A common mathematical model used in optics\cite{kogelnik_coupled-wave_1969} to compute the wavevector $\mathbf{k}_d$ of the diffrated wave from the wave vector $\mathbf{k}_p$ of the probe wave and the grating vector $\mathbf{k}_g$ is based on simple vector addition motivated by \eqref{eq:kg}, i. e.

\begin{equation} \label{eq:KVCM}
\mathbf{k}_d = \mathbf{k}_g + \mathbf{k}_p.
\end{equation}

This method is referred to as \emph{k-vector closure} by some authors. However, in the off-Bragg condition, the method leads to the physically invalid phenomenon that the length of the diffracted wave vector is different from the wave number of the probe wave. Therefore, the method can be improved\cite{uchida_calculation_1973} by enforcing that the following holds. 
\begin{equation} \label{eq:energycons}
\vert \mathbf{k}_d \vert = \vert \mathbf{k}_p \vert .
\end{equation}
Experimental validations, also of alternative theories, exist and have been discussed elsewhere \cite{prijatelj_far-off-bragg_2013}. In our discussion, we will typically assume that the HOE can be approximated locally by a volume grating and use a Cartesian coordinate system spanned by an orthonormal basis $\{\mathbf{t}_1, \mathbf{t}_2,  \mathbf{n} \}$ with the unit surface normal $\mathbf{n}$ of the HOE substrate as one coordinate axis. Then we can enforce \eqref{eq:energycons} by defining a new grating vector $\mathbf{k^{\prime}}_d$ setting
\begin{eqnarray}
\mathbf{k^{\prime}}_d \cdot \mathbf{t}_1 & = & (\mathbf{k}_g + \mathbf{k}_p)  \cdot \mathbf{t}_1 \label{eq:kdtangent1} \\
\mathbf{k^{\prime}}_d \cdot \mathbf{t}_2 & = & (\mathbf{k}_g + \mathbf{k}_p)  \cdot \mathbf{t}_2 \label{eq:kdtangent2} \\
\mathbf{k^{\prime}}_d \cdot \mathbf{n} & = & \sqrt{\vert \mathbf{k}_p \vert^2 - ((\mathbf{k}_g + \mathbf{k}_p) \cdot \mathbf{t}_1)^2 - ((\mathbf{k}_g + \mathbf{k}_p) \cdot \mathbf{t}_2)^2 } \label{eq:kdnormal}
\end{eqnarray}
provided that \eqref{eq:kdnormal} has a real solution corresponding physically to a propagating diffraction order.

Another alternative -- which is also helpful in visualizing volume diffraction in some instances -- suggests that the Bragg planes serve as partial mirrors or scatterers (due to the Fresenel conditions at boundary interfaces). Thus, in the on-Bragg condition, volume diffraction can be interpreted as reflection at the Bragg planes\cite{brotherton-ratcliffe_comparative_2014}.  

To fully predict the optical characteristics of an HOE, we need to know the direction in which the diffracted light propagates, as well as how much of the incident light is actually diffracted into this new direction. In this paper, we will limit the discussion to the simplifying assumption that there are only two propagating waves present in a volume grating. This is a common assumption in optics known as coupled wave theory\cite{kogelnik_coupled-wave_1969, uchida_calculation_1973}. An alternative approach is described by the parallel stacked mirror (PSM) methods\cite{brotherton-ratcliffe_comparative_2014}. For the interested reader, there exists a vast body of work ranging from introduction into diffraction theory for holography \cite{mihaylova_understanding_2013} to surveys of different diffraction theories and experimental validations, e.g. with the focus on the far off-Bragg regime \cite{prijatelj_far-off-bragg_2013}.  

For our discussion, we assume that an incident plane wave (the probe wave) represented by a wave vector $\mathbf{k}_p$ is diffracted at an HOE with the diffracted light propagating along the wavevector $\mathbf{k}_d$ also referred to as the first diffraction order. If the diffraction efficiency $\eta$ of the HOE is less than $1$, i.e. not all the light is diffracted into the first diffraction order, then a so-called zero diffraction order continuing along the direction of $\mathbf{k}_p$ will be present. Since for a retina scanner display we are primarily interested in phase holograms anyway, we will neglect possible absorption effects in the holographic material itself. Thus the contribution to the zero diffraction order is assumed to be $1-\eta$.

\begin{remark}
	It is noteworthy to point out that the geometric framework presented in the following is quite general and allows different levels of sophistication for the local HOE model. Since the grating vector locally describes the HOE geometry in full (assuming a volume grating approximation), we can even use rigorous methods, such as the \emph{Fourier modal method} (FMM), also known as \emph{rigorous coupled wave analysis} (RCWA)\cite{gaylord_analysis_1985}. Thus we can in principle use relatively simple scalar models based on two-wave coupled wave analysis as well as the full vectorial Maxwell equations in our modelling framework.
\end{remark}

\subsection{Example: HOE constructed from spherical wave interference} \label{subsec:sphereexpl}
We recall the complex representation of a spherical wave, i.e. an electromagnetic field where the surfaces of constant phase, i.e. the wavefronts, are contrentic spheres centered at a fixed point $\mathbf{r}_0$, to be
\begin{equation} \label{eq:spherewave}
E_{\mathbf{r}_0}(\mathbf{r},t) = A_{\mathbf{r}_0}(\mathbf{r}) \cdot e^{\mathrm{i}k|\mathbf{r} - \mathbf{r}_0]} \cdot e^{-\mathrm{i}\omega t}
\end{equation}
where the amplitude is given by 
\begin{equation} \label{eq:spherewaveamp}
A_{\mathbf{r}_0}(\mathbf{r}) = \frac{A_{0}}{\pi \cdot |\mathbf{r} - \mathbf{r}_0]}
\end{equation}
and 
\begin{equation}\label{eq:wavenumber}
k = \frac{2\pi}{\lambda}
\end{equation}
is called the \emph{wave number}.

We consider two spherical waves, one diverging from the point $\mathbf{r}_1$, and the second one converging towards the point  $\mathbf{r}_2$. Note that the fact that one wave is diverging and the other is converging is accounted for by changing the sign in the phase term arising from the scalar product of the wave vector and the spatial coordinate. Thus we have the recording waves 
\begin{equation}
E_{\mathbf{r}_1}(\mathbf{r}) = A_{\mathbf{r}_1}(\mathbf{r}) \cdot e^{\mathrm{i}k|\mathbf{r} - \mathbf{r}_1]}
\end{equation}
and
\begin{equation}
E_{\mathbf{r}_2}(\mathbf{r}) = A_{\mathbf{r}_2}(\mathbf{r}) \cdot e^{-\mathrm{i}k|\mathbf{r} - \mathbf{r}_2]}
\end{equation}
where the amplitudes and wavenumber are again defined as in \eqref{eq:spherewaveamp} and \eqref{eq:wavenumber}, respectively.

The situation is illustrated in Figure \ref{fig:sphereHOE} for the recording and reconstruction configuration of such an HOE.

\begin{figure}
	\centering
	\includegraphics[scale = 1.0]{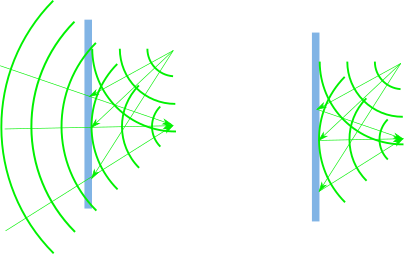}  
	\caption{Left: Hologram recording by two wave interference. Right: Wavefront reconstruction by diffraction at a holographic structure.} \label{fig:sphereHOE}
\end{figure}

The absolute value squared of interference pattern of these two spherical waves is given by
\begin{equation}
\begin{split}
\left| E_{\mathbf{r}_1}(\mathbf{r}) + E_{\mathbf{r}_2}(\mathbf{r}) \right|^2 
= & \left| A_{\mathbf{r}_1}(\mathbf{r}) \right|^2 + \left| A_{\mathbf{r}_2}(\mathbf{r}) \right|^2 \\
& + A_{\mathbf{r}_1}(\mathbf{r}) A_{\mathbf{r}_2}(\mathbf{r})  e^{-\mathrm{i}k\left(|\mathbf{r} - \mathbf{r}_2| + |\mathbf{r} - \mathbf{r}_1| \right)}  \\
& + A_{\mathbf{r}_1}(\mathbf{r}) A_{\mathbf{r}_2}(\mathbf{r})   e^{\mathrm{i}k\left(|\mathbf{r} - \mathbf{r}_2| + |\mathbf{r} - \mathbf{r}_1| \right)} \\
= & \left| A_{\mathbf{r}_1}(\mathbf{r}) \right|^2 + \left| A_{\mathbf{r}_2}(\mathbf{r}) \right|^2 \\
& + A_{\mathbf{r}_1}(\mathbf{r}) A_{\mathbf{r}_2}(\mathbf{r}) \cos \left( k \left(|\mathbf{r} - \mathbf{r}_2| + |\mathbf{r} - \mathbf{r}_1| \right) \right)
\end{split}
\end{equation}
and in an ideal holographic material is proportional to the refractive index modulation.

We observe that the oscillating (i.e. the cosine) term is constant on sets where 
\begin{equation}
k \left(|\mathbf{r} - \mathbf{r}_2| + |\mathbf{r} - \mathbf{r}_1| \right) = const.
\end{equation} 

Note that an ellipsoid of revolution $E_{\mathbf{r}_1, \mathbf{r}_2}$ with focii $\mathbf{r}_1$ and $\mathbf{r}_2$ is characterized by the fact that for any point on the boundary $\mathbf{r} \in \partial E_{\mathbf{r}_1, \mathbf{r}_2}$  the sum of the distances from $\mathbf{r}$ to the two foci  is constant, i.e. there exists a real number $C > 0$ such that 
\begin{equation}
E_{\mathbf{r}_1, \mathbf{r}_2} = \left\{ \mathbf{r} \in \mathbb{R}^3: |\mathbf{r} - \mathbf{r}_2| + |\mathbf{r} - \mathbf{r}_1| \le C  \right\}.
\end{equation}

Thus, we make the following observation.

\begin{proposition}
	The Bragg structure of a holographic optical element corresponding to two point sources -- one wave diverging, one wave converging -- located at $\mathbf{r}_1$ and $\mathbf{r}_2$, respectively, consists of a spatial variation of the index of variation where loci of constant index of refraction are segments of the boundary of ellipsoids of revolution about the axis $\overline{\mathbf{r}_1 \mathbf{r}_2}$ with focii $\mathbf{r}_1$ and $\mathbf{r}_2$. 
\end{proposition}

Note that we will also refer to these loci of constant index of refraction as \emph{Bragg isosurfaces}.

\begin{proof}
	Let $C \in \mathbb{R}$, $C>0$ be constant. Then the oscillating cosine term in the interference pattern and thus the corresponding local index of refraction in the resulting volume Bragg structure is constant on the set
	\begin{equation}
	E = \left\{ \mathbf{r} \in \mathbb{R}^3: |\mathbf{r} - \mathbf{r}_2| + |\mathbf{r} - \mathbf{r}_1|  = \frac{C}{k} \right\}
	\end{equation}
	which describes the boundary of an ellipsoid of revolution with focii $\mathbf{r}_1$ and $\mathbf{r}_2$.
\end{proof}

From the above observations, we obtain immediately the following result.
\begin{proposition}
	We consider an HOE recorded with spherical waves -- one wave diverging, one wave converging -- originating from two point sources located at $\mathbf{r}_1$ and $\mathbf{r}_2$, respectively. At each point the local grating vector $\mathbf{k}_g$ is colinear to the normal vector of an ellipsoids of revolution about the axis $\overline{\mathbf{r}_1 \mathbf{r}_2}$ with focii $\mathbf{r}_1$ and $\mathbf{r}_2$.
\end{proposition}

We point out here, that we will use planar volume gratings as local approximations for the curved Bragg isosurfaces in the remainder of this paper. This approximation is expected to be justified if the distances of construction points of the HOE to the HOE substrate are suffciently large, such that the local curvature of the curved Bragg isosurfaces are large enough to allow a local approximation by the tangent plane. It is noteworthy that the use of curved geometric entitites such as ellipsoids may prove useful for a system analysis, when the HOE properties have to be modelled not only pointwise, but for example in an open disc around a reference point. Interestingly, similar techniques, where curved geometric surfaces have been used to replace tangent planes as local approximations, have been employed already in the analysis of freeform reflector and lens design problems\cite{kochengin_determination_1997, graf_optimal_2012}. The associated partial differential equations are still attracting researchers' attention both from an analytical\cite{trudinger_local_2014, trudinger_local_2020} as well as from a numerical\cite{de_leo_numerical_2017, doskolovich_optimal_2019, romijn_inverse_2020, mingazov_use_2020} point of view. We plan to return to the topic of supporting quadric methods for HOE in a separate investigation. 

\subsection{Modelling HOE via grating vevtor fields on embedded surfaces}
Although volume holograms are also referred to as \emph{thick} holograms in the optics literature, their thickness is to be judged with respect to the wavelength of light. Thus, a phototpolymer of 10 to 20 micron can already contain an optically thick hologram. Since the lateral dimensions are typically in the range of tens of millimeters, e.g. for lenses for smartglasses, the thickness of the photopolymer is assumed to be negligible from a macroscopic geometry point of view. This leads to the following geometric model for an HOE.

\begin{definition}[HOE model and  grating vector field] \label{def:HOEmodel}
	An HOE model consists of an embedded surface $\mathcal{H}$ in $\mathbb{R}^3$ and an associated grating vector field 
	\begin{equation}
	\mathbf{k}_g: \mathcal{H} \to \mathcal{H} \times \mathbb{R}^3, \mathbf{r} \mapsto \mathbf{k}_g(\mathbf{r})
	\end{equation}
\end{definition}

\section{A formulation of the forward and inverse problem for curved HOE in terms of moving frames and grating vector fields} \label{sec:fw-prblm}

In the previous section, we have developed a framework to locally model an HOE as a volume grating using the grating vector formalism from holography and combining it with the geometric concept of a vector field. We will develop in this section a mathematical framework introducing another differential geometric concept, namely moving frames, to account for deformations of the HOE geometry after recording. 

\subsection{Basic assumptions and outline of the method}

For the sake of clearity we make the following assumptions:

\begin{enumerate}
	\item HOEs are modelled as embedded two-dimensional surfaces in $\mathbf{R}^3$, typically as the graph of a sufficiently smooth function over a fixed reference plane, such as the $(x,y)$-plane. Note that convex functions are in any case twice differentable almost everywhere\cite{busemann_convex_2013}.
	Thus, we consider primarily functions which are at least twice differentiable, i.e. of class $\mathcal{C}^2$.
	\item The physical thickness of the photopolymer is included in the local model used to compute the diffraction efficiencies and neglected in the macroscopic geometry of the HOE surface otherwise.
	\item The planar HOE is contained in the $(x,y)$-plane of a Cartesian coordiante system $(x,y,z)$ centered at the origin $\mathcal{O} = (0,0,0)$.
	\item A parametrization of the curved HOE as an embedded surface in $\mathbb{R}^3$ is given as the graph of a convex function over the $(x,y)$-plane. \label{assume:convexity}
	\item The curved HOE is a subset of a surface that is rotationally symmetric about the $z$-axis. \label{assume:cylsymmetry}      
\end{enumerate}

\begin{remark}
	Assumptions \ref{assume:convexity} and \ref{assume:cylsymmetry} allow for generalizations -- as mentioned previously -- and are motivated by the application at hand which fosuses on embedding spherically deformed HOE into ophtalmic leneses.
\end{remark}

In the following, we develop a framework to formulate Problems \ref{pblm:fw} and \ref{pblm:inv} more rigorously and establish the tools for an analytic and numerical analysis of curved HOE as well as the basis for a design method to precompensate deterministic deformations of a planar HOE during the recording process.

The core idea of our method is as follows:
\begin{enumerate}
	\item We describe the inner structure of the HOE,  i.e. the Bragg-surfaces using a grating vector model to encode the local orientation and spacing of the Bragg surfaces, cf. Definition \ref{def:HOEmodel}. We will also refer to this inner structure as the microscopic geometry of the HOE,
	\item We relate the inner structure of the HOE to the macroscopic geometry of the surface associated with the HOE using local coordinate systems which we derive from moving frames. Each frame is formed by an orthorgonal set of vectors, one being the local surface normal, the other two being tangent vectors to be chosen in a suitable manner. 
	\item At first order, deformation of the HOE surface $\mathcal{H}$ into the HOE surface $\tilde{\mathcal{H}}$ corresponds to a change of the unit tangent vectors and the unit surface normal vector $\{\mathbf{t}_1, \mathbf{t}_2,  \mathbf{n} \}$ 
	into a new local frame $\{\tilde{\mathbf{t}}_1, \tilde{\mathbf{t}}_2,  \tilde{\mathbf{n}} \}$.
	\item We assume a one-to-one coresspondence between points $\mathbf{r} \in \mathcal{H}$ and $\tilde{\mathbf{r}} \in \tilde{\mathcal{H}}$. If $\tilde{\mathcal{H}}$ is a graph of a function over the $(x,y)$-plane this condition is for example fulfilled by the orthorgonal projection along the $z$-axis, other possible projections will be discussed below.
	
	\item Thus, the grating vectors of the deformed HOE are obtained from the original grating vector 
	\begin{equation}
	\mathbf{k}_g (\mathbf{r}) = g_1(\mathbf{r}) \mathbf{t}_1(\mathbf{r})  + g_2(\mathbf{r}) \mathbf{t}_2(\mathbf{r}) + g_3(\mathbf{r})\mathbf{n}(\mathbf{r}) 
	\end{equation}
	as the vector in the frame $\{\tilde{\mathbf{t}}_1, \tilde{\mathbf{t}}_2,  \tilde{\mathbf{n}} \}$ 
	with the same coordinate values as $\mathbf{k}_g$.
	\begin{equation}
	\tilde{\mathbf{k}}_g (\tilde{\mathbf{r}}(\mathbf{r})) = g_1 (\mathbf{r}) \tilde{\mathbf{t}}_1(\tilde{\mathbf{r}}(\mathbf{r}))  + g_2 (\mathbf{r}) \tilde{\mathbf{t}}_2(\tilde{\mathbf{r}}(\mathbf{r})) + g_3 (\mathbf{r}) \tilde{\mathbf{n}}(\tilde{\mathbf{r}}(\mathbf{r})) 
	\end{equation} 	
\end{enumerate}
In essence, we use a point-wise correspondence between local frames to transform one vector field $\mathbf{k}_g$ on a surface $\mathcal{H}$ into another vector field $\tilde{\mathbf{k}}_g$ on a surface $\tilde{\mathcal{H}}$. The vector fields $\mathbf{k}_g$ and $\tilde{\mathbf{k}}_g$ serve to model locally the optical properties of the HOEs $\mathcal{H}$ and $\tilde{\mathcal{H}}$, respectively. The general idea of the method is illustrated graphically in Figure \ref{fig:kvecdef}.

\begin{figure}
	\centering
	\includegraphics[scale=0.5]{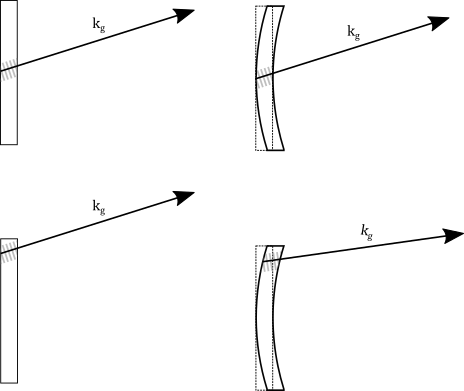}
	\caption{Accounting for influence of changes of the macroscopic geometry of an HOE via vector field transformations. The HOE is approximated locally by a volume grating with grating vector $\mathbf{k}_g$.} \label{fig:kvecdef} 
\end{figure}

\subsection{Forward problem for convex graphs with cylindrical symmetry} \label{subsec:fw-prblm}

\begin{definition}[Central projections]\label{def:centralproj}
	Let $D \subset \mathbb{R}^2$ be a bounded domain and  $\Gamma_{\tilde{h}}$ be the graph of a positive, convex function $\tilde{h}: D \to \mathbb{R}, (x,y) \mapsto \tilde{h}(x,y)$. Moreover, assume that $\Gamma_{\tilde{h}}$ is contained in a set which is rotationally symmetric about the $z$-axis. We define a set of transformations $\mathcal{T}_{c}$  from the $(x,y)$-plane onto $\Gamma_{\tilde{h}}$, whose elements have the following form.
	\begin{equation}
	T_{\mathbf{C}}:\mathbb{R}^2 \times {0} \to \Gamma_{\tilde{h}}, \mathbf{r} \mapsto \overline{\mathbf{C} \mathbf{r}} \cap \Gamma_{\tilde{h}}
	\end{equation}
	where $\mathbf{C}$ is a point on the $z$-axis called the center of projection, and $\overline{\mathbf{C} \mathbf{r}}$ defines the segment joining the center of projection $\mathbf{C}$ and the point $ \mathbf{r} \in \mathbb{R}^2 \times {0}$
\end{definition}

\begin{remark} \label{rmk:orthoproj}
	Note that the orthorgonal projection along the $z$-axis, i.e. 
	\begin{equation}
	T_{\mathbf{\infty}}:\mathbb{R}^2 \times {0} \to \Gamma_{\tilde{h}}, (x,y,0) \mapsto (x,y, \tilde{h})
	\end{equation}
	can be considered a limiting case for $\mathbf{C} = (0, 0, C)$ as $C \to \infty$. 
\end{remark}

\begin{definition}[Admissible surface transformations] \label{def:admproj}
	The set of admissible projections $\mathcal{T}_{adm}$ is defined as
	\begin{equation}
	\mathcal{T}_{a} = \mathcal{T}_{c} \cup \{ T_{\mathbf{\infty}} \}
	\end{equation}
\end{definition}

\begin{proposition} \label{prop:1to1}
	Let $T_\mathbf{C} \in \mathcal{T}_{a}$ and consider  the preimage of $\tilde{\mathcal{H}}$ denoted
	\begin{equation}
	D_\mathbf{C}  = {T_\mathbf{C}}^{-1}(\tilde{\mathcal{H}})
	\end{equation}
	Then the projection 
	\begin{equation}
	T_\mathbf{C}: D_\mathbf{C} \to \tilde{\mathcal{H}}
	\end{equation}
	is bijective.
\end{proposition}

\begin{proof}
	Let $T_\mathbf{C} \in \mathcal{T}_{a}$. Suppose that $\mathbf{C}$ is contained in the upper half-space. Suppose there is a point $(x,y,0)$ which is mapped to two points $\mathbf{P}_1 = (x^{\prime}, y^{\prime}, \tilde{h}(x^{\prime},y)^{\prime})$ and $\mathbf{P}_2 =(x^{\prime \prime}, y^{\prime \prime}, \tilde{h}(x^{\prime \prime},y)^{\prime \prime})$ on $\Gamma_{\tilde{h}}$. 
	
	We consider the projection of the line $\overline{\mathbf{P}_1 \mathbf{P}_2}$ into the $(x,y)$-plane. Note that this is a line containing the coordinate origin $(0, 0, 0)$. Therefore we can restrict $\tilde{h}$ to the line $\overline{(0,0,0), (x,y,0)}$ and parametrize it as a function of a single parameter, i.e. we have a function $\tilde{h}(s\mathrm{e}^{i \phi(x,y,0)})$ in a single variable $s$, in radial diraction from $(0,0,0)$ to $(x,y,0)$. Here, $\phi(x,y,0)$ denots the angle of the point $(x,y,0)$ in polar coordinates. 
	Morevover, the function is convex, and therefore, the segment joining $\mathbf{P}_1$ and $\mathbf{C}$ is contained in the epigraph of $\tilde{h}(r)$. Since $\mathbf{P}_1$ is also contained in this segment, we conclude that $ \tilde{h}(x^{\prime \prime},y^{\prime \prime}) > \tilde{h}(x^{\prime},y^{\prime})$. This is in contradiction to the assumptions that $\Gamma_{\tilde{h}}$ is convex and rotationally symmetric about the $z$-axis. 
\end{proof}

Next, we will introduce the notion of moving frames to define local coordinate systems associated with the macroscopic HOE geometry. 

\begin{definition}[Moving Frames] \label{def:movframes} 
	Let $\tilde{h}$ be a function whose graph $\Gamma_{\tilde{h}}$ contains an HOE $\tilde{\mathcal{H}}$. We consider the family of curves defined by 
	\begin{equation} \label{eq:framegenerators}
	c_{\phi}(s) = \{ s\mathrm{e}^{\mathrm{i} \phi} \} \times \{ \tilde{h}(s\mathrm{e}^{\mathrm{i} \phi}) \}
	\end{equation}
	where $s \in (0, \infty)$ and $\phi \in [0, 2\pi)$ denote polar coordinates in the $(x,y)$-plane.
	For a fixed $\phi \in [0, 2\pi)$ we define moving frames by the unit vectors
	\begin{eqnarray}
	\mathbf{t}_{c_\phi}(s) & = & \frac{\frac{\mathrm{d}}{\mathrm{d}s} c_{\phi}(s)}{\vert \frac{\mathrm{d}}{\mathrm{d}s} c_{\phi}(s) \vert} \label{eq:frametangent} \\
	\mathbf{b}_{c_\phi}(s) & = & \mathbf{t}_{c_\phi}(s) \times \mathbf{n}_{c_\phi}(s) \label{eq:framebinormal} \\
	\mathbf{n}_{c_\phi}(s) & = & \frac{\left( \partial_x \tilde{h}(s \mathrm{e}^{\mathrm{i} \phi}), \partial_y \tilde{h}(s \mathrm{e}^{\mathrm{i} \phi}), -1 \right)}{\vert \left( \partial_x \tilde{h}(s \mathrm{e}^{\mathrm{i} \phi}), \partial_y \tilde{h}(s \mathrm{e}^{\mathrm{i} \phi}), -1 \right) \vert} \label{eq:framenormal}
	\end{eqnarray}
\end{definition}

\begin{remark}
	With the above Definition \ref{def:movframes}, we have introduced a family of curves motivated by the underlying symmetry which we assume in the application problem. This family of curves serves to define local orthonormal bases across the HOE surface. Thus, we have introduced a field of frames\cite{cartan_differential_2006, sapiro_geometric_2006}. In the present derivation, we have not explicitly parametrized the curves by arclength. However, note that we assume the HOE surface to be bounded, i.e. there exists a compact set containing the parameter domain. Therefor, we can parametrize the curves \eqref{eq:framegenerators} by arclength. With this parametrization, equations \eqref{eq:frametangent}, \eqref{eq:framebinormal}, and \eqref{eq:framenormal} correspond to the \emph{Darboux frames} of the curves \eqref{eq:framegenerators}. 	
\end{remark}

Using the moving frames from Definition \ref{def:movframes}, we are now in the position to establish a formalism to account for changes in the microscpic geometry of an HOE due to a macroscopic deformation.

\begin{proposition}[Induced transformation of grating vector fields via moving frames] \label{prop:inducedframes}
	Let $T_\mathbf{C} \in \mathcal{T}_{adm}$. Consider a planar HOE 
	$$(\mathcal{H}, \mathbf{k}_g)$$ 
	with moving frames 
	$$\{ \mathbf{t}_{c_\phi}(s), \mathbf{b}_{c_\phi}(s), \mathbf{n}_{c_\phi}(s) \}$$  
	and a surface 
	$$\tilde{\mathcal{H}} = T_\mathbf{C}(\mathcal{H})$$ 
	with moving frames 
	$$\{ \tilde{\mathbf{t}}_{c_\phi}(s), \tilde{\mathbf{b}}_{c_\phi}(s), \tilde{\mathbf{n}}_{c_\phi}(s) \}$$
	as defined in Definition \ref{def:movframes}. Then $\mathcal{H}$ induces a grating vector field on $\tilde{\mathcal{H}}$ where the induced grating vector field $\tilde{\mathbf{k}}_g$ is defined as
	\begin{equation}
	\begin{split}
	\tilde{\mathbf{k}}_g(x,y, \tilde{h}(x,y)) = & g_1 \tilde{\mathbf{t}}_{c_\phi(x,y)}(\sqrt{x^2 + y^2}) \\ 
	& + g_2 \tilde{\mathbf{b}}_{c_\phi(x,y)}(\sqrt{x^2 + y^2}) + g_3 \tilde{\mathbf{n}}_{c_\phi(x,y)}(\sqrt{x^2 + y^2})
	\end{split}   
	\end{equation}
	where $g_1, g_2, g_3$ are the coordinates of the grating vectors $\mathbf{k}_g$ at the preimage of $(x,y, \tilde{h}(x,y))$ in the moving frames $\{ \mathbf{t}_{c_\phi}(s), \mathbf{b}_{c_\phi}(s), \mathbf{n}_{c_\phi}(s) \}$, i.e. 
	\begin{equation}
	\begin{split}
	\mathbf{k}_g({T_\mathbf{C}}^{-1}(x,y, \tilde{h}(x,y))) = & g_1 \mathbf{t}_{c_\phi(x,y)}(\vert {T_\mathbf{C}}^{-1}(x,y, \tilde{h}(x,y))  \vert) \\
	& + g_2 \mathbf{b}_{c_\phi(x,y)}(\vert {T_\mathbf{C}}^{-1}(x,y, \tilde{h}(x,y))  \vert) \\
	& + g_3 \mathbf{n}_{c_\phi(x,y)}(\vert {T_\mathbf{C}}^{-1}(x,y, \tilde{h}(x,y))  \vert)
	\end{split}
	\end{equation}
	We say that the grating vector field of the deformed HOE $(\tilde{\mathcal{H}}, \tilde{\mathbf{k}}_g)$ is induced from $(\mathcal{H}, \mathbf{k}_g)$ under the transformation $T_\mathbf{C}$.
\end{proposition}

\begin{remark}
	Note that our framework allows to further augment our model to include shrinkage or stretching effects due to chemical reactions or mechanical stresses by globally or locally rescaling the grating vector field. In general, deforming a planar HOE in to a curved surface with non-zero Gaussian curvature, such as a sphere segment, will result in (local) stretching, if the surface is deformed smoothly, i.e. avoiding wrinkles and tears. For a more accurate system model, these effects need to be better understood at the material and component level. This is part of on-going and future research which we intend to present in a subsequent publication. It appears that in some applications, e.g spherical wave HOEs for RSD combiners, these effects may be neglected in a simplified first order model, which still yields valuable insights into the optical effects arising from the changes in the macroscopic geometry alone.
\end{remark}

With the formalism developed above, we can provide an answer to Problem \ref{pblm:fw} as follows:

\begin{theorem}
	Let $(\mathcal{H}, \mathbf{k}_g)$ be a planar HOE and $T_\mathbf{C} \in \mathcal{T}_{adm}$ a transformation which maps $\mathcal{H}$ one-to-one to a surface $\tilde{\mathcal{H}}$. Let $\tilde{\mathbf{k}}_g$ be the vector field on $\tilde{\mathcal{H}}$ by $(\mathcal{H}, \mathbf{k}_g)$  and $T_\mathbf{C}$. Then, by construction, for any point $\mathbf{r} \in \tilde{\mathcal{H}}$ the local grating vector $\tilde{\mathbf{k}}_g(\mathbf{r})$ has the same length as the grating vector at its preimage. The diffraction of an incident probe wave with wave vector $\mathbf{k}_p$ can thus be computed using the local grating and surface normal vector in combination with a diffraction theory for volume gratings.
\end{theorem}

\begin{proof}
	The statement follows from Definition \ref{def:movframes} and Proposition \ref{prop:inducedframes}.
\end{proof}

In the following section, we will use the above framework to address Problem \ref{pblm:inv} in an analogous manner.

\subsection{The inverse problem for curved HOE in terms of moving frames and grating vector fields} \label{sec:inv-prblm}
We turn our attention to the inverse problem stated in Problem \ref{pblm:inv}. In our newly developed framework, Problem \ref{pblm:inv} consists of finding for a given recording geometry $\mathcal{H}$, transformation $T_{\mathbf{C}} \in \mathcal{T}_a$, and target HOE $(\tilde{\mathcal{H}}, \tilde{\mathbf{k}}_g)$ a grating vector field $\mathbf{k}_g$ such that the vector field induced by $(\mathcal{H}, \mathbf{k}_g)$ and $T_{\mathbf{C}}$ on $\tilde{\mathcal{H}}$ coincides with $\tilde{\mathbf{k}}_g$.

Based on the formalisms developed in Section \ref{subsec:fw-prblm}, we can also express an answer to Problem \ref{pblm:inv} in terms of admissible transformations and induced grating vector fields.

\begin{theorem}
	Let $(\tilde{\mathcal{H}}, \tilde{\mathbf{k}}_g)$ be a curved HOE and $T_\mathbf{C} \in \mathcal{T}_{adm}$ a transformation which maps a bounded domain $ \mathcal{H}  \subset \mathbb{R}^2 \times \{ 0 \}$ to $\tilde{\mathcal{H}}$. Let $\mathbf{k}_g$ be the vector field induced on $ \mathcal{H}$ by $(\tilde{\mathcal{H}}, \tilde{\mathbf{k}}_g)$ and $T_\mathbf{C}$. Then for any point $\mathbf{r} \in \mathcal{H}$ on the local grating vector $\mathbf{k}_g(\mathbf{r})$ has the same length as the grating vector at its preimage. Morevover, the vector field induced on  $\tilde{\mathcal{H}}$ by $(\mathcal{H}, \mathbf{k}_g)$ and $T_\mathbf{C}$ coincides with $\tilde{\mathbf{k}}_g$.
\end{theorem}

\begin{proof}
	The statement follows from Proposition \ref{prop:1to1}, Definition \ref{def:movframes} and Proposition \ref{prop:inducedframes}.
\end{proof}	

\begin{remark}
	Note that the geometric framework developed above does not provide a canonical choice for a transformation in the admissible set $\mathcal{T}_{a}$, although first numerical experiments show promising results for the orthogonal projection $T_{\infty}$ from Remark \ref{rmk:orthoproj}. Moreover, given a compact planar HOE surface $\mathcal{H}$ and a compact curved HOE surface $\tilde{\mathcal{H}}$, it is not clear if there always exist an admissible transform in  $\mathcal{T}_{adm}$ which maps $\mathcal{H}$ to $\tilde{\mathcal{H}}$  and vice verse. However, given a curved HOE surface $\tilde{\mathcal{H}}$ and an admissible transformation $T_{\mathbf{C}}$ we can always define an associated planar HOE surface $\mathcal{H} = {T_{\mathbf{C}}}^{-1}(\tilde{\mathcal{H}})$.
\end{remark}

\section{Application to HOE-embedding in prescription lenses for retina scanner AR-displays} \label{sec:applications}
In this section, we give a brief account of numerical experiments performed with a Matlab implementation of the framework developed in Section \ref{sec:fw-prblm}. For simplicity, this first prototpye only used the very basic k-vector closure method given in \eqref{eq:KVCM} to determine the wave vector of the diffracted light. Diffraction efficiencies were assumed be $100\%$, which is a very simplistic and physically inaccurate diffraction theory, which is nonetheless sometimes used  in optical simulation if only the propagation direction is of interest, or if on-Bragg diffraction is guaranteed by the system setup, in order to simplify the implementation; see for example the hologram surfaces implemented in Zemax OpticStudio. However, more advanced diffraction theories will be inlcuded in a more elaborate Python implementation which is currently under development.

\subsection{Qualitative experimental validation: Curving a plane wave HOE (volume grating) for lens-embedding}
Figure \ref{fig:matlabHOE} shows a visualization of the simulation results of the diffraction of a plane wave at a planar and two curved HOEs. The planar HOE was recorded using two plane waves. One recording wave was indicent at $0 \deg$ with respect to the optical axis perpendicular to the HOE at its center point, the second recording wave was incident at an angle of about $65 \deg$ with respect to the optical axis. The left column of figures shows the propagation directions along the diffracted wave vectors (red) for plane wave illumination at the recording wavelength with a probe wave (green) incident at $65 \deg$ to the $z$-axis (the optical axis). The top graphic shows a close up of the HOE surface, part of the grating vector field (blue, rescaled for visualization) is also visible. The center column shows the results for an HOE recorded in a curved phototpolymer. Aside from the macrsocopic geometry, there is no difference to the planar HOE. The right column shows the result for an HOE that was deformed after recording. In particular the view in the bottom graphics -- where the propagation directions of the diffracted wavevectors have been scaled up for visualization purposes -- shows a significant change of the optical characteristics. Instead of a plane wave, the light is diffracted into a converging wave which also shows an astigmatic focal region between 6 and 8 cm along the $z$-axis.

\begin{figure}
	\centering
	\includegraphics[width=1.0\linewidth]{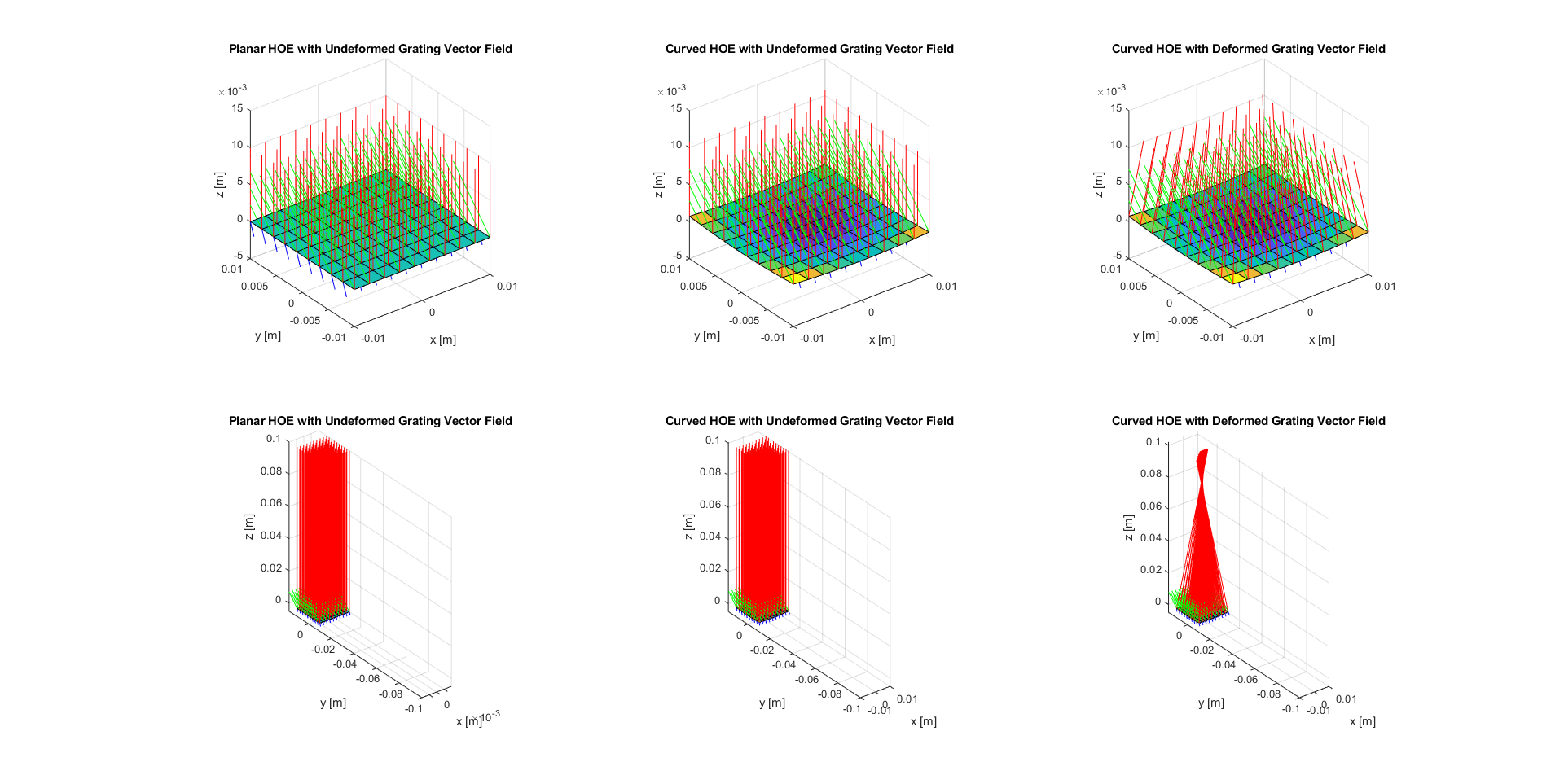}
	\caption{Simulation results for HOE reconstruction after different recording scenarios: Left: Diffracted wave (red) of a plane probe wave (green) incident at $65 \deg$ with respect to the surface normal at a planar HOE in close-up (Top) and an alternative view (Bottom). Center: Diffracted wave (red) of a plane probe wave (green) incident at $65 \deg$ with respect to the surface normal at the vertex of a spherically deformed HOE recored after the deformation in close-up (Top) and an alternative view (Bottom). Right: Diffracted wave (red) of a plane probe wave (green) incident at $65 \deg$ with respect to the surface normal at the vertex of a spherically deformed HOE recored prior to deformation in close-up (Top) and an alternative view (Bottom) showing a focusing effect and, at close inspection, an astigmatic focal region.} \label{fig:matlabHOE}
\end{figure}

For a first qualitative validation, the simulation results were compared with a simple measurement of a comparable volume grating that had been curved and embedded in a lens previously. Due to the available measurement setup, the HOE was replayed with a plane wave at $0\deg$ incidence with respect to the optical axis, i. e. in reverse compared to the simulation in Figure \ref{fig:matlabHOE}. The simulation results and the camera measurements of the cross section of the diffracted beam are shown in Figure \ref{fig:labHOE} and show both the same qualitative behavior with an astigmatic focal region in the diffracted beams.  

\begin{figure}
	\centering
	\includegraphics[width=1.0\linewidth]{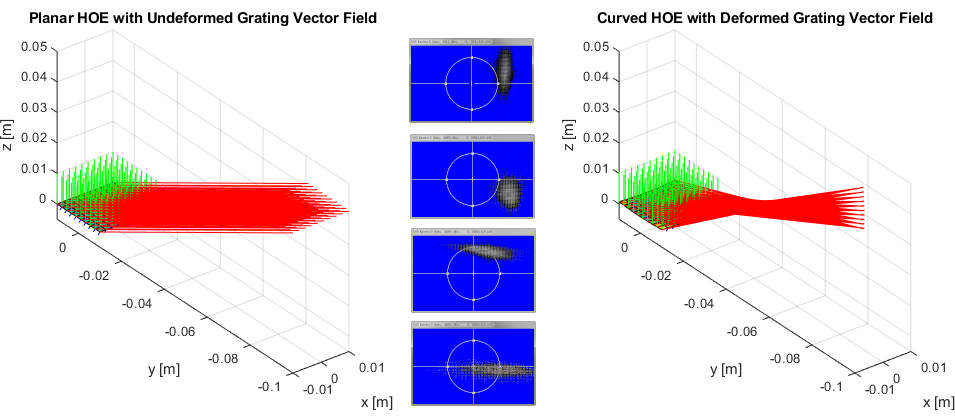} 
	\caption{Left: Simulation results for a planar volume grating showing a plane incident wave (green) and a plane diffracted wave (red); Right: Simulation results for a deformed volume grating (spherical deformation) showing a plane incident wave (green) and a diffracted wave (red) showing an astigmatic focal region. Center: Camera images showing the transverse intensity profile of the wave diffracted at acurved HOE along the propagation axis in a laboratory setup. The camera images at least qualitatively confirm the theoretically predicted astigmatism. (The measurements were facilitated by Stefanie Hartmann at CR/AEL2)} \label{fig:labHOE}
\end{figure}

\subsection{Consequences for HOE combiners in retina scanner displays recorded with spherical wave interference}
We end this section with abrief outlook on how the developed methods can be applied to the optical design of holographic combiners for RSDs. 

Figure \ref{fig:matlabRSDHOE} shows in the top sequence of graphics the simulation results for an HOE combiner recorded with a divergent and a convergent spherical wave, where the convergent wave converged towards a point source located on the z-axis through the HOE center. The left graphic shows the corresponding grating vector field (rescaled for visualization) in blue. the center graphic shows the directions of propagation of the local wave vectors of the diffracted wave when probed on-Bragg with the divergent recording wave. The right graphic shows the intersection points of the propagation directions of the diffracted wave vectors with different parallel detector planes which are distinguished by colour. 

The bottom sequence in Figure \ref{fig:matlabRSDHOE} shows the analogous simulation results for an HOE curved after recording. While the deformation of the grating vector field may not appear obvious at first glance, the propagation directions  after diffraction already show an impact on the focal region which is closer to the HOE vertex than before. Moreover, the visualization suggests again an astigmatic focal region. This impression is also confirmed by the interception points in the graphics on the right-hand side.  

\begin{figure}
	\centering
	\includegraphics[width=1.0\linewidth]{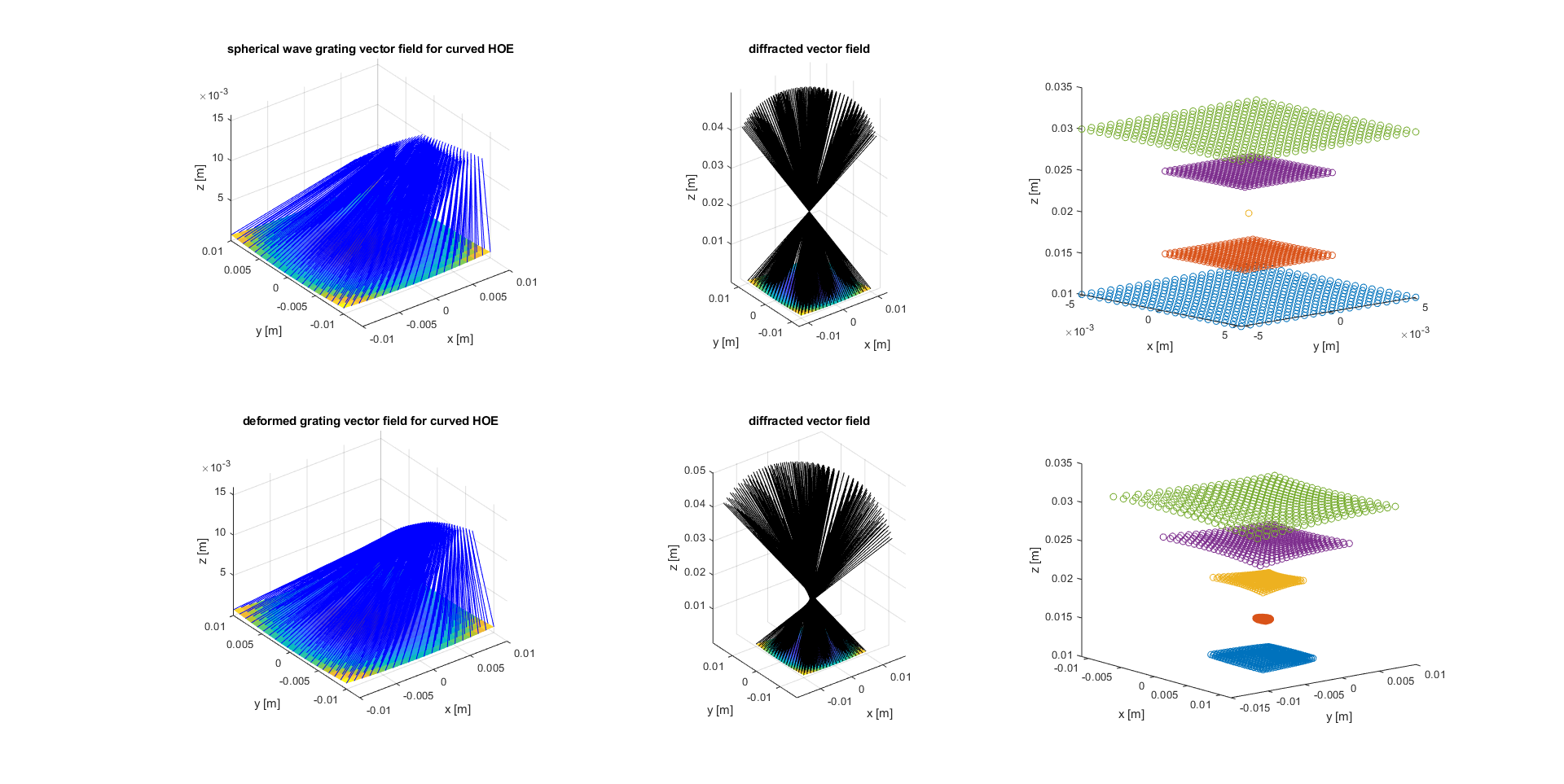} 
	\caption{Simulation resutls for a planar HOE (Top) and a curved HOE deformed after recording (Bottom): The left graphics show the grating vector fields (blue) rescaled for visualization purposes. The center graphics show the directions of propagation of the diffracted wave vectors (black), we observe that the focal region of the curved HOE is shifted in $z$-direction and astigmatic. The graphics on the right show point of interception of the propagaton directions of the diffract wave vectors with five parallel detector planes, confirming the focal shift and astigmatism for the curved HOE.} \label{fig:matlabRSDHOE}
\end{figure}

From the observations on Figure \ref{fig:matlabRSDHOE} we conclude, that the geometric framework developed in this paper is expected to prove helpful in the simulation and design of holographic combiners for RSDs which need to be curved after recording. Our method is expected to provide a systematic framework for design and optimization tools, e. g. to precompensate in the design process changes in the diffraction properties of an HOE due to deformations of the HOE geometry after recording.

As an outlook for future applications of this approach, we mention here that based on the above results, the Matlab simulation was extended to implement a parametric optimization for curved HOE created from spherical wave interference which were used successfully in a demonstrator setup of a RSD. Thus, we are confident that the framework presented in this paper offers a promising modelling and simulation approach for the optical design of curved HOE for RSD and will be investigated further.

\section{Conclusion} \label{sec:conclusion}
We have developed a geometric framework that is suitable for the investigation of optical effects arising due to a smooth deformation of the macroscopic HOE geometry after the recording process. We want to point out that the present investigation was developed independently and differs from previous work on curved HOE \cite{bang_curved_2019} by the fact that we do not focus in our investigation to cylinder segments as curved HOE surfaces. Recent work\cite{bang_curved_2019} focuses on ways to exploit the change in diffraction characteristics due to curvature in new system designs, in particular, the very interesting example of dynamic exit pupil expansion in a retina scanner display is discussed. The method presented in Section \ref{sec:fw-prblm} can be adapted to surfaces that are not rotationally symmetric to include for example cylindric deformation where the curved HOE still have zero Gaussian curvature as discussed in previous work \cite{bang_curved_2019}. However, the focus of our study was motivated by the need for a better understanding of HOEs deformed into geometries with nonzero Gaussian curvature after recording. In particular, we are interested in simulation methods to predict these effects and develop design methods to precompensate them e.g. during HOE recording. The assumption of the geometry to be rotationally symmetric is a reasonable constraint for spherical lenses. However, we are convinced that even moderate freeform lenses, e.g. due to corrections of astigmatisms, can be treated with our method. More importantly, however, we would like to point out that these generalizations are only required if the HOE is attached to an optically active surface, i.e. a boundary interface between two different media. Our method is suitable without modification even for not rotationally symmetric lenses, as long as the HOE is curved in a rotationally symmetric way. This is consistent with existing embedding processes\cite{korner_casting_2018}, where the curved HOE is fully embedded in the surrounding lens material. Thus, the assumption of rotational symmetry for this purely mechanical interface geometry does not in general limit the symmetry (or lack thereof) properties of the lens surfaces.  

As we already pointed out earlier, the present investigation does not describe the geometry of the wavefronts diffracted at the HOE as has been don in earlier work\cite{verboven_aberration_1986, peng_nonparaxial_1986} in order to derive aberration coefficients with respect to some reference wavefront for the case where the HOE is recorded on a curved substrate directly. However, our approach yields a method to describe how the optical characteristics of an HOE are affected, if its macroscopic geometry is changed after recording. Thus, we believe that the current approach may be combined with previous work on wavefront aberrations, since it yields the local ray directions for the curved HOE and thus should allow at least a numerical approach to a wavefront aberration calculation as discussed in earlier work\cite{verboven_aberration_1986, peng_nonparaxial_1986}.

The methodology presented in this paper has already been adapted successfully to the optical design and simulation of curved HOE prototypes used in table top demonstrators. Moreover, the framework opens new possibilities for further modelling and simulation aspects associated with the optical design and system architecture of retina scanner displays, which we plan to include in future investigations. To develop such methods, the author is convinced that modelling and simulation methods in optical engineering will profit from suitable mathematical tools, in particular from differential geometry and partial differential equations, to allow for variable levels of model complexity and fidelity depending on the requirements of the engineering task. On the other hand, the development of suitable mathematical models and numerical algorithms is also expected to profit from insights into the underlying physical phenomena and the optical design freedom as well as inherent constraints of novel optical components and materials.   

\section*{Acknowledgments}
The author would like to acknowledge the assistance of Dr. Stefanie Hartmann in facillitating the laboratory experiments. Further thanks go to the current and former holography and display colleagues of the Optoelectronics and Optical Systems Group at Bosch Corporate Research, the optical engineering colleagues at Bosch Sensortec GmbH, the organizers and participants of the Optimization and Simulation Workshop at Bosch Corporate Research, and to Hendrik Jansen and Dominik Hayd at PEA4-Fe for fruitful discussions on holography, display technology, as well as mathematical modelling and simulation methods. 

This manuscript has been accpeted in revised form for publication in SPIE \emph{Optical Engineering}.

\section*{Citation}
Tobias Graf, "Curved holographic optical elements from a geometric view point," Opt. Eng. 60(3) 035102 (4 March 2021) \href{https://doi.org/10.1117/1.OE.60.3.035102}{https://doi.org/10.1117/1.OE.60.3.035102}

\section*{Copyright notice}
Copyright 2021 Society of Photo-Optical Instrumentation Engineers. One print or electronic
copy may be made for personal use only. Systematic reproduction and distribution, duplication of any
material in this paper for a fee or for commercial purposes, or modification of the content of the paper
are prohibited.

\bibliography{Bib_curvedHOE}   
\bibliographystyle{spiejour}   

\end{spacing}
\end{document}